\documentclass[a4 paper, 12pt]{article}

\usepackage{arxiv}

\usepackage[utf8]{inputenc} % allow utf-8 input
\usepackage[T1]{fontenc}    % use 8-bit T1 fonts
\usepackage{hyperref}       % hyperlinks
\usepackage{url}            % simple URL typesetting
\usepackage{booktabs}       % professional-quality tables
\usepackage{amsmath, amsfonts, amssymb, amsthm, mathtools}       % blackboard math symbols
\usepackage{nicefrac}       % compact symbols for 1/2, etc.
\usepackage{microtype}      % microtypography
\usepackage{lipsum}
\usepackage{graphicx}
\usepackage{wrapfig}
\usepackage{subcaption}
\usepackage{tikz}
\usepackage[inline]{enumitem}
\usepackage{float}
\usepackage{algorithm}
\usepackage{algpseudocode}

\usepackage[backend=bibtex,style=alphabetic,natbib=true, maxbibnames=99]{biblatex}
\usepackage[autostyle=true]{csquotes}

\newtheorem{theorem}{Theorem}[section]

\newtheorem{lemma}[theorem]{Lemma}

\theoremstyle{definition}
\newtheorem{definition}[theorem]{Definition}

\theoremstyle{remark}

\def\nat{\mathbb N}

\def\nonnegreal{\mathbb R_{\ge 0}}
\def\positivereal{\mathbb R_{> 0}}

\def\rat{\mathbb Q}

\def\int{\mathbb Z}
\def\nonnegint{\mathbb Z_{\ge 0}}
\def\trans{\mathsf{T}}
\DeclareMathOperator{\supp}{supp}
\def\cc#1{\mathtt{#1}}
\def\P{\ensuremath{\cc{P}}}

\DeclareMathOperator{\NP}{\cc{NP}}

\DeclareMathOperator{\PPAD}{\cc{PPAD}}
\usepackage{xspace}
\def\problem#1{\textsc{#1}\xspace}
\def\ThreeSAT{\problem{3-SAT}}
\def\BimatrixGame#1#2{$\langle #1, #2 \rangle$-\problem{Bimatrix Game}}
\def\PlanarBimatrixGame#1#2{$\langle #1, #2 \rangle$-\problem{Planar Bimatrix Game}}
\def\UndominatedOutRegular#1#2{$\langle #1, #2 \rangle$-\problem{Undominated OutRegular Subgraph}}
\def\UndominatedOutRegularPlanar#1#2{$\langle #1, #2 \rangle$-\problem{Undominated OutRegular Subgraph on Planar Graph}}

\title{NP-hardness of Computing Uniform Nash Equilibria on Planar Bimatrix Games}
\shorttitle{On Computing Uniform Nash Equilibria on Planar Bimatrix Games}

\author{
 Takashi Ishizuka\\
  Graduate School of Mathematics,\\
  Kyushu University,\\
  744 Motooka, Nishi-ku, Fukuoka, Japan\\
  \texttt{ishizuka.takashi.664@s.kyushu-u.ac.jp} \\
   \And
 Naoyuki Kamiyama\\
  Institute of Mathematics for Industry,\\
  Kyushu University,\\
  744 Motooka, Nishi-ku, Fukuoka, Japan\\
  \texttt{kamiyama@imi.kyushu-u.ac.jp} \\
}

\begin{document}
\maketitle
\begin{abstract}
	We study the complexity of computing a uniform Nash equilibrium on a non-win-lose bimatrix game. It is known that such a problem is $\NP$-complete even if a bimatrix game is win-lose (Bonifaci et al., 2008).
	Fortunately, if a win-lose bimatrix game is planar, then uniform Nash equilibria always exist. We have a polynomial-time algorithm for finding a uniform Nash equilibrium of a planar win-lose bimatrix game (Addario-Berry et al., 2007). The following question is left: How hard to compute a uniform Nash equilibrium on a planar non-win-lose bimatrix game?
	This paper resolves this issue. We prove that the problem of deciding whether a non-win-lose planar bimatrix game has uniform Nash equilibrium is also $\NP$-complete.
\end{abstract}

% keywords can be removed

\section{Introduction}
\label{SectionIntroduction}

The research interaction between Game Theory and Theoretical Computer Science has helped to study the computational issues underlying fundamental game-theoretical notions.
A seminal topic of these research streams is to clarify the complexity of computing Nash equilibria in multi-player non-cooperative games.

The results of the last two decades unravel the complexity of computing Nash equilibria. It has been known that the complexity of the problem of finding one Nash equilibrium on a finite strategic-form game belongs between $\P$ and $\NP$. In particular, Papadimitriou \cite{Pap94a} introduced the complexity class $\PPAD$ and proved that such a problem is $\PPAD$-complete. 
Unfortunately, it is also $\PPAD$-complete to compute a Nash equilibrium even if the number of players is constant. Daskalakis et al.\ \cite{DGP09} have shown the $\PPAD$-completeness of Nash equilibrium computation on a three-player game. Chen and Deng \cite{CD06} have proven that the problem of computing a Nash equilibrium for a two-player strategic-form game is $\PPAD$-complete.

The most important result is the intractability for computing Nash equilibria for two-player strategic-form games such that both players gain only a unit reward at most \cite{AKV05}. Although the Lemke-Howson algorithm \cite{LH64} is known as the best algorithm for finding a Nash equilibrium, it has been known that there is an example that that algorithm has an exponential worst-case running time \cite{SS04}.

Motived from the above negative facts, some studies have concentrated on the complexity of computing specific classes of equilibria, such as pure, mixed, or correlated equilibria \cite{FPT04, Pap05}.
This paper focuses on the uniform Nash equilibria, i.e., Nash equilibria consisting of strategies in which strategies are played according to a uniform distribution. A uniform Nash equilibrium is a kind of mixed equilibrium. However, unlike mixed equilibria, uniform Nash equilibria do not always exist.
Bonifaci et al.\ \cite{BIL08} have shown that it is $\NP$-complete to decide whether a given two-player game has a uniform Nash equilibrium even if every element of payoff matrices is either $0$ or $1$. Note that we call such a game a win-lose bimatrix game.

These hardness results were based on a graph-theoretical approach. There is a one-to-one correspondence between a win-lose bimatrix game and a bipartite digraph. Addario-Berry et al.\ \cite{AOV07} have shown that a win-lose bimatrix game corresponding to a planar bipartite digraph always has a uniform Nash equilibrium, and we can find it in polynomial time.

\subsection{Our Results}
\label{SectionOurResults}
This paper sharpens the tractability-intractability boundary for the problem of computing uniform Nash equilibria on bimatrix games (not necessarily win-lose). We focus on a bimatrix game that corresponds to an edge-weighted planar bipartite digraph; we call such a game a planar bimatrix game. Furthermore, we discuss the following two unsolved issues:
\begin{enumerate}[label = (\arabic*)]
	\item Does a planar bimatrix game always have a uniform Nash equilibrium?
	\item Can we efficiently compute a uniform Nash equilibrium on a planar bimatrix game (find out if it exists)?
\end{enumerate}
This paper presents the negative answers to the above issues. We prove that the problem of deciding whether a planar bimatrix game has uniform Nash equilibria is $\NP$-complete --- this result implies that uniform Nash equilibria do not always exist. More precisely, we state that the problem of computing a uniform Nash equilibrium on a non-win-lose planar bimatrix game is also $\NP$-complete.

\paragraph{Technical Overview}
Our proof shown in Section \ref{SectionPlanarBimatrixGames} is inspired by the proof techniques used in \cite{AKV05,BIL08,CS05}.
We extend their graph-theoretical approach to one on edge-weighted digraphs; we present a one-to-one correspondence between bimatrix games and edge-weighted bipartite digraphs and characterize a uniform Nash equilibrium by a graph-theoretical notion. 
The key technique in proving the main result is the construction of the clause-variable gadget, shown in Figure \ref{FigureClauseVariableGadget}.

\section{Preliminaries}
\label{SectionPreliminaries}
Firstly, we define basic notations.
We denote by $\nat$, $\rat$, $\nonnegint$, $\nonnegreal$, and $\positivereal$ the sets of positive integers, rational numbers, non-negative integers, non-negative real numbers, and positive real numbers, respectively.
We use $[n] := \{ 1, 2, \dots, n \}$ for $n \in \nat$.
For a given digraph $G$, we denote by $V(G)$ and $E(G)$ the sets of vertices and directed edges on $G$, respectively.
For sets $A$ and $B$, we denote by $A \sqcup B$ the disjoint union of $A$ and $B$.

\subsection{Bimatrix Games}
\label{SectionBimatrixGames}
A bimatrix game is a two-player strategic form game where both players have a finite set of strategies. We call the first (resp. second) player the {\it row} (resp. {\it column}) player.
Here, we denote by $R$ and $C$ the sets of strategies for the row and column player, respectively.
A bimatrix game is specified by non-negative real matrices $M_{R}, M_{C} \in \nonnegreal^{R \times C}$; the rows and columns of both matrices are indexed by the pure strategies of the players.

Let $\rho_{R}$ and $\rho_{C}$ be positive integers. 
We call a bimatrix game $(M_{R}, M_{C})$ a $\langle \rho_{R}, \rho_{C} \rangle$-bimatrix game if there are subsets $\mathcal{A} \subseteq \rat \setminus \{ 0 \}$ with $| \mathcal{A} | = \rho_{R}$ and $\mathcal{B} \subseteq \rat \setminus \{ 0 \}$ with $| \mathcal{B} | = \rho_{C}$ such that $M_{R} \in \left( \mathcal{A} \sqcup \{ 0 \} \right)^{R \times C}$, and $M_{C} \in \left( \mathcal{B} \sqcup \{ 0 \} \right)^{R \times C}$.
Namely, a win-lose bimatrix game is a $\langle 1, 1 \rangle$-bimatrix game.

A mixed strategy is a probability distribution over pure strategies, i.e., a vector $x_{R} \in \nonnegreal^{R}$ such that $\sum_{r_i \in R} x_{R}(r_i) = 1$ is a mixed strategy of the row player. Similarly, a vector $x_{C} \in \nonnegreal^{C}$ such that $\sum_{c_j \in C} x_{C}(c_j) = 1$ is a mixed strategy of the column player.
For a mixed strategy $x$, the support $\supp(x)$ of $x$ is the set of pure strategies $i$ such that $x(i) > 0$.
A mixed strategy $x$ is said to be {\it uniform} if for every $i \in \supp(x)$, $x(i) = 1 / |\supp(x)|$.

When the row player plays a mixed strategy $x_R$ and the column player plays a mixed strategy $x_C$, the expected payoffs for the row player and the column player is $x_R^{\trans} M_R x_C$ and $x_{R}^{\trans} M_C x_C$, respectively.
A Nash equilibrium of the bimatrix game $(M_R, M_C)$ is a pair of mixed strategies $(x_R^*, x_C^*)$ satisfying for all mixed strategies $x_R \in \nonnegreal^R$ of the row player, $(x_R^*)^{\trans} M_R x_C^* \ge x_R^{\trans} M_R x_C^*$, and for all mixed strategies $x_C \in \nonnegreal^C$ of the column player, $(x_R^*)^{\trans} M_C x_C^* \ge (x_R^*)^{\trans} M_C x_C$.
We call a Nash equilibrium $(x_R^*, x_C^*)$ a uniform Nash equilibrium if both mixed strategies are uniform.

\subsection{Edge-Weighted Digraphs Induced by Bimatrix Games}
\label{SectionDigraphfromBimatrixGame}
We now describe the one-to-one correspondence between bimatrix games and edge-weighted bipartite digraphs.
Note that our notation and definition used in this paper are based on \cite{AOV07, BIL08}; these papers only dealt with unweighted digraphs.
We generalize their notion to edge-weighted digraphs.

\paragraph{One-to-One Correspondence}
Let $(M_R, M_C)$ be a bimatrix game, where the sets of strategies for the row player and the column player are $R$ and $C$, respectively.
Now, we define the edge-weighted bipartite digraph $G = (V, E, w)$ induced by the bimatrix game $(M_R, M_C)$.
Here, $w : E \to \positivereal$ represents weights on edges.
The set $V$ of vertices is $R \sqcup C$. We have an arc $(r_i, c_j)$ if $M_R(r_i, c_j) > 0$, and the weight on the arc $(r_i, c_j)$ is equal to $M_R(r_i, c_j)$. Similarly, we have an arc $(c_j, r_i)$ if $M_C(r_i, c_j) > 0$, and the weight on the arc $(c_j, r_i)$ is equal to $M_C(r_i, c_j)$.

The bimatrix game $(M_R, M_C)$ is a planar bimatrix game if the edge-weighted digraph $G$ is induced by $(M_R, M_C)$ is a planar graph.
Furthermore, we call an edge-weighted bipartite digraph $G = (R \sqcup C, E, w)$ a $\langle \rho_{r}, \rho_{c} \rangle$-edge-weighted bipartite digraph if there are two sets $\mathcal{A} \subseteq \rat \setminus \{ 0 \}$ with $| \mathcal{A} | = \rho_{r}$ and $\mathcal{B} \subseteq \rat \setminus \{ 0 \}$ with $| \mathcal{B} | = \rho_{C}$ such that each edge from $R$ to $C$ has the weight within $\mathcal{A}$ and each edge from $C$ to $R$ has the weight within $\mathcal{B}$.

Remark that it is possible that the reverse transformation. Thus, we can construct a bimatrix game from an edge-weighted bipartite digraph.

Let $G = (V, E, w)$ be an edge-weighted bipartite digraph.
Since $G$ is bipartite, we can naturally separate the set of vertices $V$ into two disjoint subsets $R$ and $C$.
We regard $R$ and $C$ as the sets of pure strategies for the row and column players, respectively.
The payoff matrix of the row player $M_{R}$ is defined as follows:
For each $r_{i} \in R$, $M_{R}(r_{i}, c_{j}) = w(r_{i}, c_{j})$ if $(r_{i}, c_{j}) \in E$; otherwise, $M_{R}(r_{i}, c_{j}) = 0$.
The definition of the payoff matrix of the column player $M_{C}$ is similar, i.e., for each $c_{j} \in C$, $M_{C}(r_{i}, c_{j}) = w(c_{j}, r_{i})$ if $(c_{j}, r_{i}) \in E$; otherwise, $M_{C}(r_{i}, c_{j}) = 0$.

From the above observation, we obtain the following statement:
\begin{lemma}
\label{LemmaCorrespondenceBetweenBimatrixGame&Digraph}
	For every bimatrix game, there is a corresponding edge-weighted bipartite digraph; vice versa.
\end{lemma}

\paragraph{Undominated Out-regular Subgraphs}
For a finite subset $S \subseteq V$, let $G[S]$ be the subgraph induced by $S$.
We denote by $\Gamma^{+}(v)$ the set of out-neighbors of a vertex $v \in V$, and by $\Gamma^{+}_{S}(v)$ the set of out-neighbors of a vertex $v \in V$ that are in $S \subseteq V$.

Fix arbitrary non-empty subset of strategies $S = S_R \sqcup S_C \subseteq V$, where $\emptyset \neq  S_{R} \subseteq R$ and $\emptyset \neq S_{C} \subseteq C$.
The induced subgraph $G[S]$ is $(\alpha, \beta)$ out-regular if the following two conditions (i) for every $r_i \in S_R$, $\sum_{c_j \in \Gamma^{+}_{S}(r_i)} w(r_i, c_j) = \alpha$; (ii) for every $c_j \in S_C$, $\sum_{r_i \in \Gamma^{+}_{S}(c_j)} w(c_j, r_i) = \beta$.
We say that $G$ has an out-regular subgraph if there is a triple $(\alpha, \beta, S)$ such that $G[S]$ is $(\alpha, \beta)$ out-regular.

A vertex $v \in V \setminus S$ dominates an $(\alpha, \beta)$ out-regular subgraph $G[S]$ if it satisfies that eighter $\sum_{c_j \in \Gamma^{+}_{S}(v)} w(v, c_j) > \alpha$ if $v \in R$; or $\sum_{r_i \in \Gamma^{+}_{S}(v)} w(v, r_i) > \beta$ if $v \in C$.
We say that an $(\alpha, \beta)$ out-regular subgraph $G[S]$ is undominated if there are no vertices that dominate $G[S]$.
We say that $G$ has an undominated out-regular subgraph if there is a triple $(\alpha, \beta, S)$ such that $G[S]$ is an undominated $(\alpha, \beta)$ out-regular subgraph.

\begin{theorem}
\label{TheoremEquivalence}
	A bimatrix game has a uniform Nash equilibrium if and only if the corresponding edge-weighted bipartite digraph has an undominated out-regular subgraph; vice versa.
\end{theorem}
\begin{proof}
	Let $(M_{R}, M_{C})$ be a bimatrix game, where the sets of strategies for the row player and the column player are $R$ and $C$, respectively.
	We denote by $G$ the edge-weighted bipartite digraph corresponding to $(M_{R}, M_{C})$.
	Fix a subset of pure strategies $S = S_{R} \sqcup S_{C}$, where $\emptyset \neq S_{R} \subseteq R$ and $\emptyset \neq S_{C} \subseteq C$.
	
	Now, we show that $G[S]$ is an undominated out-regular subgraph if and only if the pair of uniform strategies $(x_{R}, x_{C})$ such that $\supp(x_{R}) = S_{R}$ and $\supp(x_{C}) = S_{C}$ is a Nash equilibrium of $(M_{R}, M_{C})$.
	
	First, we assume that $G[S]$ is an undominated $(\alpha, \beta)$ out-regular subgraph form some positive reals $\alpha$ and $\beta$. Then, we show that the pair of uniform strategies $(x_{R}, x_{C})$ defined as $\supp(x_{R}) = S_{R}$ and $\supp(x_{C}) = S_C$ is a Nash equilibrium.
	When the column player plays the strategy $x_{C}$, the row player gains the expected payoff $\alpha / |S_{C}|$ by playing $r_{i} \in S_{R}$, but gains the expected payoff at most $\alpha / |S_{C}|$ by playing $r_{i'} \in R \setminus S_{R}$. Therefore, the uniform strategy is a best response to $x_{C}$ for the row player.
	Similarly, we can see that the uniform strategy $x_{C}$ is a best response to $x_{R}$ for the column player.
	Thus, $(x_{R}, x_{C})$ is a uniform Nash equilibrium.
	
	Next, we assume that the pair of uniform strategies $(x_{R}, x_{C})$ such that $\supp(x_{R}) = S_{R}$ and $\supp(x_{C}) = S_{C}$ is a Nash equilibrium. For the sake of contradiction, we suppose that there are strategies $r_{i}, r_{k} \in S_R$ such that $\sum_{c_{j} \in \Gamma_{S}^{+}(r_{i})} w(r_{i}, c_{j}) > \sum_{c_{\ell} \in \Gamma_{S}^{+}(r_{k})} w(r_{k}, c_{\ell})$ when the column player plays $x_{C}$.
	In this case, the row player gains the expected payoff $(1 / |S_{C}|) \sum_{c_{j} \in \Gamma_{S}^{+}(r_{i})} w(r_{i}, c_{j})$ by playing the strategy $r_{i}$. On the other hand, she gains the expected payoff $(1 / |S_{C}|) \sum_{c_{\ell} \in \Gamma_{S}^{+}(r_{k})} w(r_{k}, c_{\ell})$ by playing $r_{k}$.
	From our assumption, the row player will get a lager expected payoff for playing $r_{i}$ than for playing $r_{k}$, which contradicts that $(x_{R}, x_{C})$ is a Nash equilibrium. Therefore, $G[S]$ is an $(\alpha, \beta)$ out-regular subgraph for some positive reals $\alpha$ and $\beta$.
	Using a similar technique, it is easy to see that $G[S]$ is undominated. 
\end{proof}

Now, we show an important property of undominated out-regular subgraphs, which we use to prove the main result.
If the edge-weighted bipartite digraph $G = (V, E, w)$ is strongly connected, and a subset $S$ induces an undominated out-regular subgraph, then we have $\sum_{u \in \Gamma_{S}^{+}(v)} w(v, u) > 0$ for each $v \in S$.

\begin{lemma}
	\label{LemmaImportantProperty}
	Assume that the edge-weighted bipartite digraph $G = (V, E, w)$ is strongly connected, and a subset $S \subseteq V$ induces 	the undominated out-regular subgraph $G[S]$.
	Then, every vertex $v \in S$ satisfies that $\sum_{u \in \Gamma_{S}^{+}(v)} w(v, u) > 0$.
\end{lemma}
\begin{proof}
	For the sake of contradiction, we assume that there is a vertex $v \in S$ such that	$\sum_{u \in \Gamma_{S}^{+}(v)} w(v, u) = 0$.
	We take a vertex $u \in S$ that belongs to the different part from $v$.
	Since the digraph $G$ is strongly connected, there is an arc $(v', u) \in E$.
	If $v' \in S$, it contradicts that $S$ induces an out-regular subgraph. On the other hand, if $v \in V \setminus S$, it contradicts that $G[S]$ has no dominating vertices.
\end{proof}

\subsection{A Brief Introduction to Computational Complexity}
\label{SectionComputationalComplexity}
We begin with a brief overview of various computational complexity-theoretical notions that will be needed for the remainder of this paper.
Let $\Sigma = \{ 0, 1\}$ be a binary alphabet. We denote by $\Sigma^{*}$ the set of finite strings.
A decision problem $L \subseteq \Sigma^{*}$ is defined as follows: given a string $x$, decide whether $x$ is in $L$.
A complexity class $\mathcal{C}$ is a set of decision problems. The complexity class $\NP$ is the set of all decision problems that are decidable by a non-deterministic Turing machine.

A reduction is a key concept of Computational Complexity Theory. A reduction from the problem $A$ to the problem $B$ is a function $f: \Sigma^{*} \to \Sigma^{*}$ such that a string $x \in A$ if and only if  $f(x) \in B$.
When such a function $f$ is computable by a deterministic Turing machine, we say that such a reduction is a polynomial-time reduction.
Furthermore, we say that two decision problems $A$ and $B$ are polynomially equivalent if there are polynomial-time reductions from $A$ to $B$ and from $B$ to $A$.
A decision problem $L$ is $\NP$-hard if all decision problems belonging to $\NP$ are polynomial-time reducible to $L$. A decision problem $L$ is $\NP$-complete if $L$ belongs to $\NP$ and is $\NP$-hard.
Note that the problem \ThreeSAT (see Definition \ref{Def:ProblemTreeSAT}) is the well-known $\NP$-complete problem \cite{Coo71,Lev73}.

\paragraph{The List of Decision Problems}
We conclude this brief introduction by resenting a list of decision problems that will be used in this paper.

\begin{definition}
\label{Def:ProblemTreeSAT}
	\ThreeSAT:

	\noindent
	\textsc{\bf Input:}
	\begin{itemize}
		\item a $3$-CNF formula $\phi: \{0, 1\}^{n} \to \{0, 1\}$.
	\end{itemize}
	
	\noindent 
	\textsc{\bf Task:} Decide whether
	\begin{itemize}
		\item there exists an assignment $x \in \{0, 1\}^{n}$ such that $\phi(x) = 1$.
	\end{itemize}
\end{definition}

\begin{definition}
	\BimatrixGame{\rho_{r}}{\rho_{c}}
	
	\noindent
	\textsc{\bf Input:}
	\begin{itemize}
		\item a $\langle \rho_{r}, \rho_{c} \rangle$-bimatrix game $(M_R, M_C)$
	\end{itemize}
	
	\noindent 
	\textsc{\bf Task:} Decide whether
	\begin{itemize}
		\item there is a uniform Nash equilibrium $(x_R, x_C)$ of $(M_R, M_C)$.
	\end{itemize}
\end{definition}

\begin{definition}
	\PlanarBimatrixGame{\rho_{r}}{\rho_{c}}
	
	\noindent
	\textsc{\bf Input:}
	\begin{itemize}
		\item a $\langle \rho_{r}, \rho_{c} \rangle$-planar bimatrix game $(M_R, M_C)$
	\end{itemize}
	
	\noindent 
	\textsc{\bf Task:} Decide whether
	\begin{itemize}
		\item there is a uniform Nash equilibrium $(x_R, x_C)$ of $(M_R, M_C)$.
	\end{itemize}
\end{definition}

\begin{definition}
	\UndominatedOutRegular{\rho_{r}}{\rho_{c}}
	
	\noindent
	\textsc{\bf Input:}
	\begin{itemize}
		\item a $\langle \rho_{r}, \rho_{c} \rangle$ edge-weighted bipartite digraph $G = (V, E, w)$
	\end{itemize}
	
	\noindent 
	\textsc{\bf Task:} Decide whether
	\begin{itemize}
		\item $G$ has an undominated out-regular subgraph.
	\end{itemize}
\end{definition}

\begin{definition}
	\UndominatedOutRegularPlanar{\rho_{r}}{\rho_{c}}
	
	\noindent
	\textsc{\bf Input:}
	\begin{itemize}
		\item a $\langle \rho_{r}, \rho_{c} \rangle$ edge-weighted bipartite planar digraph $G = (V, E, w)$
	\end{itemize}
	
	\noindent 
	\textsc{\bf Task:} Decide whether
	\begin{itemize}
		\item $G$ has an undominated out-regular subgraph.
	\end{itemize}
\end{definition}

\section{On the Complexity of $\langle 1, 2 \rangle$-Planar Bimatrix Games}
\label{SectionPlanarBimatrixGames}
The purpose of this section is to prove the $\NP$-hardness of the problem of deciding whether a given $\langle 2, 2 \rangle$-planar bimatrix game has a unique Nash equilibrium. So, we prove the next theorem.

\begin{theorem}
	\label{TheoremNP-comletenss4PlanarBimatrixGame}
	The problem \PlanarBimatrixGame{1}{2} is $\NP$-complete.
\end{theorem}

It is trivial that the problem \PlanarBimatrixGame{1}{2} belongs to the class $\NP$. Thus, it suffices to prove that this problem is $\NP$-hard. For this, we present a polynomial-time reduction from \ThreeSAT to \PlanarBimatrixGame{1}{2}.

Recall that Theorem \ref{TheoremEquivalence} states that the existence of uniform Nash equilibria for a bimatrix game is equal to the existence of undominated out-regular subgraphs of the corresponding edge-weighted bipartite digraph. Moreover, the reduction between a bimatrix game and an edge-weighted bipartite digraph is polynomial-time computable and preserves planar. Therefore, these two problems \PlanarBimatrixGame{1}{2} and \UndominatedOutRegularPlanar{1}{2} are polynomially equivalent. 

To complete the proof of Theorem \ref{TheoremNP-comletenss4PlanarBimatrixGame}, we present a polynomial-time reduction from \ThreeSAT to \UndominatedOutRegular{1}{2}, i.e., we show the following lemma:

\begin{lemma}
	\label{LemmaReductionThreeSAT-UndominatedOutRegular}
	\ThreeSAT is polynomial-time reducible to \UndominatedOutRegularPlanar{1}{2}.
\end{lemma}

Here, we sketch the proof of Lemma \ref{LemmaReductionThreeSAT-UndominatedOutRegular}; the full proof of this lemma can be found in Section \ref{SectionProofLemmaReductionThreeSAT-UndominatedOutRegular}.
In the first step, we reduce the problem \ThreeSAT to \UndominatedOutRegular{1}{2} in Section \ref{SectionFirstStep}.
For each \ThreeSAT instance $\phi: \{ 0, 1 \}^n \to \{ 0, 1 \}$, we construct the $\langle 1, 2 \rangle$-edge-weighted bipartite digraph $G_{\phi}$. Note that $G_{\phi}$ is not planar; it may have some crossing points.
Hence, we eliminate all crossing points from the graph $G_{\phi}$ in the second step. The key technique to eliminate crossing points is the construction of the clause-variable gadget, shown in Figure \ref{FigureClauseVariableGadget}. Denote by $H_{\phi}$ the planar digraph formed by transforming $G_{\phi}$.
To complete the proof of Lemma \ref{LemmaReductionThreeSAT-UndominatedOutRegular}, we prove that $G_{\phi}$ has an undominated out-regular subgraph if and only if $H_{\phi}$ also has it.

\section{Proof of Lemma \ref{LemmaReductionThreeSAT-UndominatedOutRegular}}
\label{SectionProofLemmaReductionThreeSAT-UndominatedOutRegular}

\subsection{Reduction from \ThreeSAT to \UndominatedOutRegular{1}{2}}
\label{SectionFirstStep}
Let $\phi : \{ 0, 1 \}^n \to \{ 0, 1 \}$ be a $3$-CNF formula.
Without loss of generality, we assume that the $3$-CNF formula $\phi$ holds the following three properties:
\begin{itemize}
	\item There are $m$ clauses.
	\item Each clause has exactly three literals.
	\item Each literal is contained in at least one clause.
\end{itemize}
This subsection constructs an edge-weighted bipartite digraph $G_{\phi}$ such that every undominated out-regular subgraph leads to a satisfying assignment of $\phi$. We show how to embed the digraph $G_{\phi}$ into a plane in Section \ref{SectionSecondStep}.

The set of vertices of $G_{\phi}$ consists of the followings:
\begin{itemize}
	\item for each $i \in [n]$, variable vertices $x_i$, $\bar{x}_{i}$ and variable coordinating vertices $y_i$, $\bar{y}_{i}$, $z_{i_1}$, $z_{i_2}$,
	\item for each $j \in [m]$, clause vertex $C_{j}$ and clause coordinating vertices $v_{j_1}$, $v_{j_2}$, $v_{j_3}$, $u_{j_1}$, $u_{j_2}$, $u_{j_3}$, and
	\item an additional vertex $a$.
\end{itemize}

\noindent
We separate these vertices into two parts: 
\begin{itemize}
	\item the first part $\mathcal{V}_{R} = \{ y_{i}, \bar{y}_{i}, z_{i_1} ~;~ i \in [n] \} \cup \{ C_{j}, u_{j_1}, u_{j_2}, u_{j_3} ~;~ j \in [m] \}$ and
	\item the second part $\mathcal{V}_{C} = \{ x_{i}, \bar{x}_{i}, z_{i_2} ~;~ i \in [n] \} \cup \{ v_{j_{1}}, v_{j_2}, v_{j_3} ~;~ j \in [m] \} \cup \{ a \}$.
\end{itemize}

First, we describe the arcs going out from clause vertices and clause coordinating vertices.
For a $j \in [m]$, we denote by $\ell_{j_1}$, $\ell_{j_2}$, $\ell_{j_3}$ the three literals contained in the clause $C_{j}$.
For each $k \in \{ 1, 2, 3 \}$, we create the following three arcs: $(C_{j}, v_{j_k})$, $(v_{j_{k}}, u_{j_k})$, and $(u_{l_k}, \ell_{j_k})$.
The weights of these arcs are defined as $w(C_{j}, v_{j_k}) = 1$, $w(v_{j_{k}}, u_{j_k}) = m+n$, and $w(u_{l_k}, \ell_{j_k}) = 1$.
We illustrate an example of our construction in Figure \ref{FigureClauseGadget}.

Second, we describe the arcs going out from the additional vertex $a$.
For every $j \in [m]$, we create the arc $(a, C_{j})$, and define its weight as $w(a, C_{j}) = 1$.
Furthermore, for each $i \in [n]$, we create the arc $(a, z_{i_1})$, and define its weight as $w(a, z_{i_1}) = 1$.

Third, we describe the arcs going out from the variable vertices and variable coordinating vertices.
For each $i \in [n]$, we create the arcs $(x_{i}, y_{i})$, $(y_{i}, a)$, $(\bar{x}_{i}, \bar{y}_{i})$, $(\bar{y}_{i}, a)$, $(a, z_{i_1})$, $(z_{i_1}, z_{i_2})$, $(z_{i_2}, y_{i})$, and $(z_{i_2}, \bar{y}_{i})$. The weights of these arcs are defined as $w(x_{i}, y_{i}) = m+n$, $w(y_{i}, a) = 1$, $w(\bar{x}_{i}, \bar{y}_{i}) = m+n$, $w(\bar{y}_{i}, a) = 1$, $w(a, z_{i_1}) = 1$, $w(z_{i_1}, z_{i_2}) = 1$, $w(z_{i_2}, y_{i}) = m+n$, and $w(z_{i_2}, \bar{y}_{i}) = m+n$. 
We illustrate the construction of arcs for every $i \in [n]$ in Figure \ref{FigureConstructionGadgets}; we call this construction a variable gadget.

Finally, we embed the graph $G_{\phi}$ into a $3m \times 2mn$ grid graph. We arrange all clause coordinating vertices (the type of $u_{j_k}$) to the left and all variable vertices to the bottom. Each variable vertex has $m$ points to connect clause coordinating vertices. At the variable vertices, we label each point to connect the adjacent edge with $1, 2, \dots, m$.
Suppose that the arc $u_{j_k}$ to $x_{i}$ exists. This arc extends horizontally from $u_{j_k}$ to just above the $j$-th connection point of $x_{i}$. The arc is then lowered vertically to the $j$-th point of $x_{i}$. We repeat this procedure to connect every directed edge between clause coordinating vertices and variable vertices. We lexicographically chose a connecting point of a clause vertex without loss of generality.
We illustrate an example of our construction in Figure \ref{FigureExampleSATGraph}.

\begin{figure}
	\centering
	\begin{tikzpicture}[scale = 1.0]
	\tikzset{clause/.style={draw, rectangle, text centered}};
	\tikzset{literal/.style={draw, circle, text centered}};
	
	\node[literal] at (-2, 2) (a) {$a$};
	\node[clause] at (0, 2) (c) {$C_j$};
	\node[literal] at (2, 0) (v1) {$v_{j_1}$};
	\node[literal] at (2, 2) (v2) {$v_{j_2}$};
	\node[literal] at (2, 4) (v3) {$v_{j_3}$};
	\node[clause] at (4, 0) (u1) {$u_{j_1}$};
	\node[clause] at (4, 2) (u2) {$u_{j_2}$};
	\node[clause] at (4, 4) (u3) {$u_{j_3}$};
	\node[literal] at (6, 0) (l1) {$\ell_{j_1}$};
	\node[literal] at (6, 2) (l2) {$\ell_{j_2}$};
	\node[literal] at (6, 4) (l3) {$\ell_{j_3}$};

	\draw[->, thick, blue] (a) -- (c);
	\draw[->, thick, blue] (c) -- (v1);
	\draw[->, thick, blue] (c) -- (v2);
	\draw[->, thick, blue] (c) -- (v3);
	\draw[->, dashed, red] (v1) -- (u1);
	\draw[->, dashed, red] (v2) -- (u2);
	\draw[->, dashed, red] (v3) -- (u3);
	\draw[->, thick, blue] (u1) -- (l1);
	\draw[->, thick, blue] (u2) -- (l2);
	\draw[->, thick, blue] (u3) -- (l3);
	
\end{tikzpicture}
	\caption{The Construction of arcs for every $j \in [m]$. Here, weights on \textcolor{red}{dashed} and \textcolor{blue}{thick} edges are \textcolor{red}{$m+n$} and \textcolor{blue}{$1$}, respectively.}
	\label{FigureClauseGadget}
\end{figure}
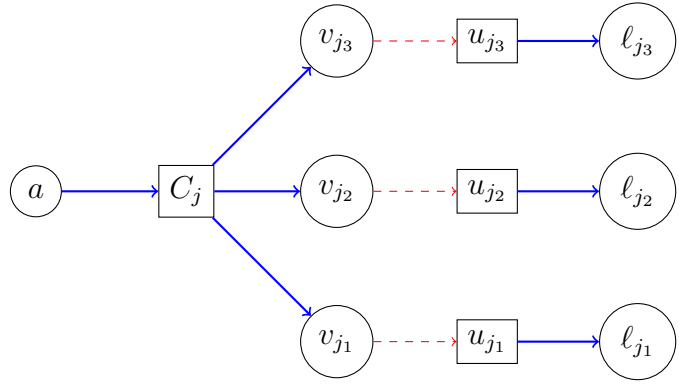

\begin{figure}
	\begin{tikzpicture}[scale = 0.6]
	\tikzset{clause/.style={draw, rectangle, text centered}};
	\tikzset{literal/.style={draw, circle, text centered}};

	\node[clause] at (0, 1) (c11) {$u_{1_1}$};
	\node[clause] at (0, 2) (c12) {$u_{1_2}$};
	\node[clause] at (0, 3) (c13) {$u_{1_3}$};

	\node[clause] at (0, 4) (c21) {$u_{2_1}$};
	\node[clause] at (0, 5) (c22) {$u_{2_2}$};
	\node[clause] at (0, 6) (c23) {$u_{2_3}$};

	\node[clause] at (0, 7) (c31) {$u_{3_1}$};
	\node[clause] at (0, 8) (c32) {$u_{3_2}$};
	\node[clause] at (0, 9) (c33) {$u_{3_3}$};
	
	% C_1
	\draw[->, thick, blue] (c11) -- (1, 1) -- (1, 0);
	\draw[->, thick, blue] (c12) -- (10, 2) -- (10, 0);
	\draw[->, thick, blue] (c13) -- (13, 3) -- (13, 0);
	
%	\draw[-, thick] (0, 0.8) -- (0, 3.2) -- (-1, 3.2) -- (-1, 0.8) -- (0, 0.8);
%	\node at (-0.5, 2) {$C_1$};
	
	% C_2
	\draw[->, thick, blue] (c21) -- (8, 4) -- (8, 0);
	\draw[->, thick, blue] (c22) -- (17, 5) -- (17, 0);
	\draw[->, thick, blue] (c23) -- (20, 6) -- (20, 0);
	
%	\draw[-, thick] (0, 3.8) -- (0, 6.2) -- (-1, 6.2) -- (-1, 3.8) -- (0, 3.8);
%	\node at (-0.5, 5) {$C_2$};
	
	% C_3
	\draw[->, thick, blue] (c31) -- (6, 7) -- (6, 0);
	\draw[->, thick, blue] (c32) -- (12, 8) -- (12, 0);
	\draw[->, thick, blue] (c33) -- (24, 9) -- (24, 0);
	
%	\draw[-, thick] (0, 6.8) -- (0, 9.2) -- (-1, 9.2) -- (-1, 6.8) -- (0, 6.8);
%	\node at (-0.5, 8) {$C_3$};
	
	% Variables
	% x_1
	\draw[-, thick] (0.8, 0) -- (3.2, 0) -- (3.2, -1) -- (0.8, -1) -- (0.8, 0);
	\node at (2, -0.5) {$x_1$};
	\draw[-, thick] (3.8, 0) -- (6.2, 0) -- (6.2, -1) -- (3.8, -1) -- (3.8, 0);
	\node at (5, -0.5) {$\bar{x}_1$};
	
	% x_2
	\draw[-, thick] (6.8, 0) -- (9.2, 0) -- (9.2, -1) -- (6.8, -1) -- (6.8, 0);
	\node at (8, -0.5) {$x_2$};
	\draw[-, thick] (9.8, 0) -- (12.2, 0) -- (12.2, -1) -- (9.8, -1) -- (9.8, 0);
	\node at (11, -0.5) {$\bar{x}_2$};
	
	% x_3
	\draw[-, thick] (12.8, 0) -- (15.2, 0) -- (15.2, -1) -- (12.8, -1) -- (12.8, 0);
	\node at (14, -0.5) {$x_3$};
	\draw[-, thick] (15.8, 0) -- (18.2, 0) -- (18.2, -1) -- (15.8, -1) -- (15.8, 0);
	\node at (17, -0.5) {$\bar{x}_3$};
	
	% x_4
	\draw[-, thick] (18.8, 0) -- (21.2, 0) -- (21.2, -1) -- (18.8, -1) -- (18.8, 0);
	\node at (20, -0.5) {$x_4$};
	\draw[-, thick] (21.8, 0) -- (24.2, 0) -- (24.2, -1) -- (21.8, -1) -- (21.8, 0);
	\node at (23, -0.5) {$\bar{x}_4$};
	
\end{tikzpicture}
	\caption{The digraph constrcuted from a given \ThreeSAT instance $\phi(x) = (x_1 \vee \bar{x}_{2} \vee x_3) \wedge (x_2 \vee \bar{x}_{3} \vee x_{4} ) \wedge ( \bar{x}_{1} \vee x_3 \vee \bar{x}_{4} )$, where $C_1 =  x_1 \vee \bar{x}_{2} \vee x_3$, $C_2 = x_2 \vee \bar{x}_{3} \vee x_{4}$, and $C_3 = \bar{x}_{1} \vee x_3 \vee \bar{x}_{4}$.}
	\label{FigureExampleSATGraph}
\end{figure}
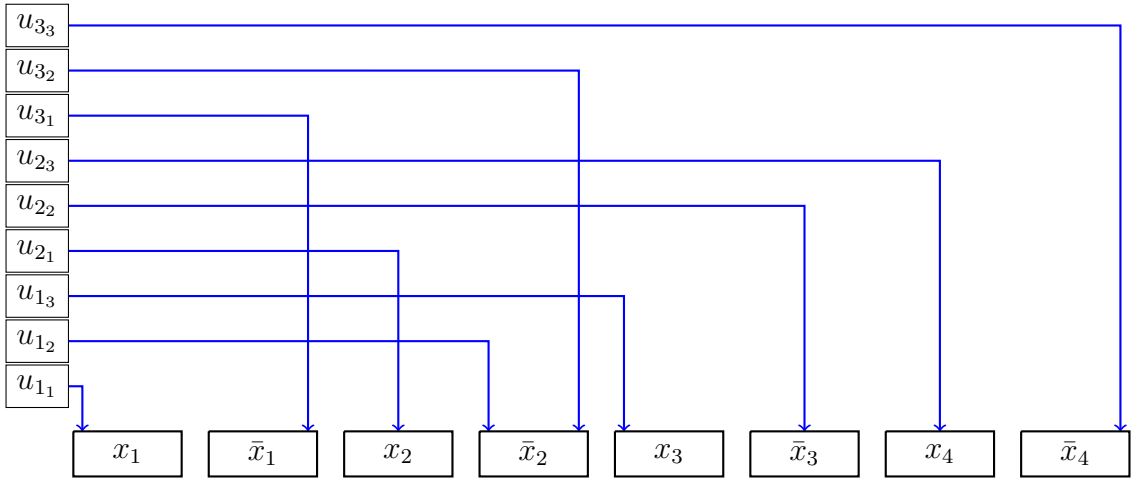

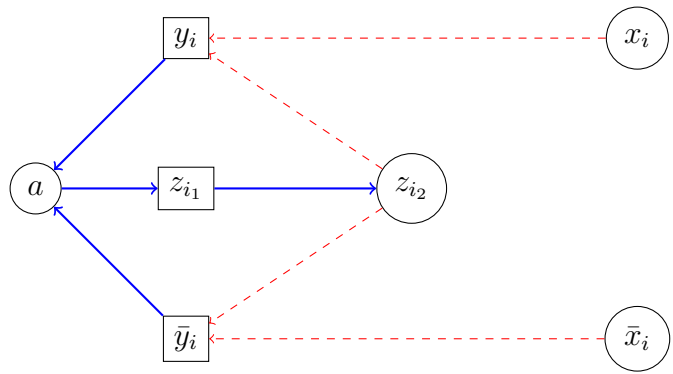
\begin{figure}
	\centering
	\begin{tikzpicture}[scale = 1.0]
	\tikzset{clause/.style={draw, rectangle, text centered}};
	\tikzset{literal/.style={draw, circle, text centered}};
	
	\node [literal] at (-4, 2) (a) {$a$};
	\node [clause] at (-2, 2) (z) {$z_{i_1}$};
	\node [literal] at (1, 2) (w) {$z_{i_2}$};
	\node [literal] at (4, 4) (tx) {$x_i$};
	\node [literal] at (4, 0) (nx) {$\bar{x}_{i}$};
	\node [clause] at (-2, 4) (ty) {$y_i$};
	\node [clause] at (-2, 0)  (ny) {$\bar{y}_{i}$};
	
	\draw[->, thick, blue] (a) -- (z); %node [midway, above] {$1$};
	\draw[->, thick, blue] (z) -- (w); %node [midway, above] {$1$};
	\draw[->, dashed, red] (w) -- (ty); %node [midway, right] {$m+n$};
	\draw[->, dashed, red] (w) -- (ny); %$node [midway, right] {$m+n$};
	\draw[->, dashed, red] (tx) -- (ty); %node [midway, above] {$m+n$};
	\draw[->, thick, blue] (ty) -- (a); %node [midway, above] {$1$};
	\draw[->, dashed, red] (nx) -- (ny); %node [midway, above] {$m+n$};
	\draw[->, thick, blue] (ny) -- (a); %node [midway, above] {$1$};
	
\end{tikzpicture}
	\caption{The construction of the variable gadget. Here, weights on \textcolor{red}{dashed} and \textcolor{blue}{thick} edges are \textcolor{red}{$m+n$} and \textcolor{blue}{$1$}, respectively.}
	\label{FigureConstructionGadgets}
\end{figure}

We complete the construction of $G_{\phi}$. 
From the construction, we can straightforwardly see the following three facts:
\begin{enumerate}[label = (\roman*)]
	\item Every arc is either going from $\mathcal{V}_{R}$ to $\mathcal{V}_{C}$ or from $\mathcal{V}_{C}$ to $\mathcal{V}_{R}$.
	\item The weight on an arc outgoing from a vertex in $\mathcal{V}_{R}$ is one.
	\item The weight on an arc outgoing from a vertex in $\mathcal{V}_{C}$ is either one or $m + n$.
\end{enumerate}
Thus, the digraph $G_{\phi}$ is a $\langle 1, 2 \rangle$-edge-weighted bipartite digraph.
Hence, the corresponding two-player game is a $\langle 1, 2 \rangle$-bimatrix game.
Now, we prove the following theorem. 

\begin{theorem}
	\label{TheoremSATiffUndominatedOutregular}
	A Boolean formula $\phi$ has a satisfying assignment if and only if the graph $G_{\phi}$ has an undominated out-regular subgraph.
\end{theorem}
\begin{proof}
	First, we assume that $\phi$ is satisfiable, i.e., there exists an assignment $\xi \in \{ 0, 1 \}^{n}$ such that $\phi(\xi) = 1$.
	Then, we construct the subset of vertices $S \subseteq V(G_{\phi})$ such that 
	\begin{enumerate}
		\item for each $i \in [n]$, $x_{i}, y_{i} \in S$ if $\xi_{i} = 1$; otherwise, $\bar{x}_{i}, \bar{y}_{i} \in S$
		\item for each $i \in [n]$, $z_{i_1}, z_{i_2} \in S$
		\item for each $j \in [m]$, $C_{j} \in S$
		\item for each $j \in [m]$, we select exactly one literal $\ell_{j_{k}} \in S$ contained in $C_{j}$, and put clause coordinating vertices $v_{j_k}$ and $u_{j_k}$ into $S$.
		\item the additional vertex $a$ in $S$.
	\end{enumerate}

	What remains is to prove that the subset $S$ induces an undominated out-regular subgraph of $G_{\phi}$.
	
	First, we show that $G_{\phi}[S]$ is out-regular. Each vertex in $S \cap \mathcal{V}_{R}$ has exactly one out-neighbor, and the corresponding arc has the weight of one. Each vertex in $S \cap \mathcal{V}_{C}$ except the additional vertex $a$ has exactly one out-neighbor, and the corresponding arc has the weight of $m + n$. The additional vertex $a$ has $m + n$ out-neighbors, and each arc has the weight of one, which implies that the total weight is equal to $m + n$. Therefore, $G[S]$ is a $(1, m+n)$ out-regular induced subgraph.
	
	Next, we prove that $G_{\phi}[S]$ is undominated. Note that the only vertex that may dominate $G_{\phi}[S]$ is $z_{i_2}$ for any $i \in [n]$. For each $i \in [n]$, the variable coordinating vertex $z_{i_2}$ does not dominate $G_{\phi}[S]$ since $S$ contains at most one of $x_{i}$ and $\bar{x}_{i}$ that are out-neighbors of $z_{i_2}$. Hence, $G_{\phi}[S]$ is an undominated $(1, m+n)$ out-regular induced subgraph.
	
	Therefore, the graph $G_{\phi}$ has an undominated out-regular induced subgraph if the $3$-CNF formula $\phi$ has a satisfiable assignment.
	
	Next, we assume that the subset $S$ of the vertices on the digraph $G_{\phi}$ induces an undominated out-regular subgraph.
	Before we describe the construction of a satisfiable assignment of $\phi$, we observe the properties of the subset $S$. More specifically, we prove the following properties:
	\begin{itemize}
		\item the additional vertex $a \in S$,
		\item for each $j \in [m]$, $C_{j} \in S$,
		\item for each $i \in [n]$, $z_{i_1}, z_{i_2} \in S$,
		\item for each $i \in [n]$, at most one of $x_{i}$ and $\bar{x}_{i}$ is contained in $S$.
	\end{itemize}
	
	Note that the digraph $G_{\phi}$ is strongly connected, and every directed edge has a weight of at least one. 
	Therefore, we have $\sum_{u \in \Gamma_{S}^{+}(v)} w(v, u) > 0$ for every vertex $v \in S$ from Lemma \ref{LemmaImportantProperty}.
	Furthermore, we can see that the additional vertex $a$ is contained in $S$. The graph given by removing the additional vertex $a$ from the digraph $G_{\phi}$ is an acyclic digraph. Note that any induced subgraph $H[S]$ on an acyclic digraph $H$ has a vertex $v$ such that $\Gamma_{S}^{+}(v) = \emptyset$, which implies that $\sum_{u \in \Gamma_{S}^{+}(v)} w(v, u) = 0$.
	
	The subset $S$ contains at least one out-neighbor of the additional vertex $a$.
	When a clause vertex $C_{j}$ is contained in $S$, there exists exactly one $k \in \{ 1, 2, 3 \}$ such that $S$ contains the set of clause coordinating vertices $\{ v_{j_k}, u_{j_k} \}$. We write $x_{i}$ that is the out-neighbor of $u_{j_k}$, without loss of generality. From Lemma \ref{LemmaImportantProperty}, $x_{i}$ and the variable coordinating vertex $y_{i}$ are  also contained in $S$.
	On the other hand, when a variable coordinating vertex $z_{i_1}$ is contained in $S$, the variable coordinating vertex $z_{i_1}$ and at least one of its out-neighbor are also contained in $S$ from Lemma \ref{LemmaImportantProperty}.
	These facts implies that all out-neighbors of $a$ are contained in $S$ since the number of $a$'s out-neibhgors is $m + n$, which is equals to the weight on arcs $(x_{i}, y_{i})$ and $(z_{i_2}, y_{i})$.
	
	Note that for each $i \in [n]$, the subset $S$ contains at most one of variable coordinating vertices $y_{i}$ and $\bar{y}_{i}$. For the sake of contradiction, we assume that there is $i \in [n]$ such that both $y_{i}$ and $\bar{y}_{i}$ are contained in $S$. In this case, the variable coordinating vertex $z_{i_2}$ dominates the subset $S$ because the total weight of $x_{i}$ is $m + n$ but the total weight of $z_{i_2}$ is $2m + 2n$.
	Hence, for each $i \in [n]$, at most one of $x_{i}$ and $\bar{x}_{i}$ is contained in $S$. If the variable vertex $x_{i}$ is in $S$, the variable coordinating vertex $y_{i}$ is also in $S$ from Lemma \ref{LemmaImportantProperty}. Similarly, if $\bar{x}_{i}$ is in $S$, $\bar{y}_{i}$ is also in $S$.

	From now on, we describe how to construct a satisfiable assignment of $\phi$. We can easily find such an assignment $\xi \in \{ 0, 1 \}^{n}$. For each $i \in [n]$, we define $\xi_{i} = 1$ if and only if $x_{i} \in S$.
	
	What remains is to prove that the assignment $\xi$ holds that $\phi(\xi) = 1$. 
	It is easy to see that every clause vertex is contained in the subset $S$. For each clause vertex $C_{j}$, exactly one literal $\ell_{j_k}$ is in $S$. By the definition of $\xi$, if the literal $\ell_{j_k}$ corresponds to a variable vertex $x_{i}$, then $\xi_{i} = 1$; otherwise, $\xi_{i} = 0$. This implies that the assignment $\xi$ is a satisfiable assignment of $\phi$ since every clause contains at least one TRUE assignment.
\end{proof}

\subsection{Embedding into plane}
\label{SectionSecondStep}
In this section, we embed the graph $G_{\phi}$ constructed in Section \ref{SectionFirstStep} into a plane. Thus, we eliminate crossing points.
From simple observation, the only directed edges that may not be planar are those between the variable vertices and the clause coordinating vertices.
To eliminate every crossing point, we construct a clause-variable gadget shown in Figure \ref{FigureClauseVariableGadget} and insert this gadget into each crossing point.

\begin{figure}
	\centering
	\begin{tikzpicture}
	\tikzset{clause/.style={draw, rectangle, text centered}};
	\tikzset{literal/.style={draw, circle, text centered}};
		
	\node [literal] at (0, 0) (a) {$\epsilon$};
	
	\node [clause] at (0, 2) (ct1) {$\alpha_1$};
	\node [literal] at (0, 4) (vt1) {$\alpha_2$};
	\node [clause] at (0, 6) (ct2) { };
	
	\node [clause] at (-2, 0) (cl1) {$\beta_1$};
	\node [literal] at (-4, 0) (vl1) {$\beta_2$};
	\node [clause] at (-6, 0) (cl2) { };
	
	\node [clause] at (-2, -2) (cb11) {$\gamma_1^1$};
	\node [literal] at (-2, -4) (vb21) {$\gamma_2^1$};	
	\node [clause] at (-1, -2) (cb12) {$\gamma_1^2$};
	\node [literal] at (-1, -4) (vb22) {$\gamma_2^2$};
	\node [clause] at (1, -2) (cb1m) {$\gamma_1^{m+n}$};
	\node [literal] at (1, -4) (vb2m) {$\gamma_2^{m+n}$};
	\node at (0, -2) {$\cdots$};
	\node at (0, -3) {$\cdots$};
	\node at (0, -4) {$\cdots$};	
	\node [clause] at (0, -6) (cb3) {$\gamma_3$};
	
	\node [clause] at (2, 2) (cr11) {$\delta_1^1$};
	\node [literal] at (4, 2) (vr21) {$\delta_2^1$};
	\node [clause] at (2, 1) (cr12) {$\delta_1^2$};
	\node [literal] at (4, 1) (vr22) {$\delta_2^2$};
	\node [clause] at (2, -1) (cr1m) {$\delta_1^{m+n}$};
	\node [literal] at (4, -1) (vr2m) {$\delta_2^{m+n}$};
	\node at (2, 0) {$\vdots$};
	\node at (3, 0) {$\vdots$};
	\node at (4, 0) {$\vdots$};	
	\node [clause] at (6, 0) (cr3) {$\delta_3$};

	\draw [->, thick, blue] (ct2) -- (vt1);
	\draw [->, dashed, red] (vt1) -- (ct1);
	\draw [->, thick, blue] (ct1) -- (a);
	
	\draw [->, thick, blue] (cl2) -- (vl1);
	\draw [->, dashed, red] (vl1) -- (cl1);
	\draw [->, thick, blue] (cl1) -- (a);
	
	\draw [->, thick, blue] (a) -- (cb11);
	\draw [->, thick, blue] (a) -- (cb12);
	\draw [->, thick, blue] (cb11) -- (vb21);
	\draw [->, thick, blue] (cb12) -- (vb22);
	\draw [->, dashed, red] (vb21) -- (cb3);
	\draw [->, dashed, red] (vb22) -- (cb3);
	\draw [->, thick, blue] (a) -- (cb1m);
	\draw [->, thick, blue] (cb1m) -- (vb2m);
	\draw [->, red, dashed] (vb2m) -- (cb3);
	\draw [->, thick, blue] (vl1) -- (cb11);
	\draw [->, blue, thick] (cb3) -- (0, -7);
	
	\draw [->, thick, blue] (a) -- (cr11);
	\draw [->, thick, blue] (a) -- (cr12);
	\draw [->, thick, blue] (cr11) -- (vr21);
	\draw [->, thick, blue] (cr12) -- (vr22);
	\draw [->, dashed, red] (vr21) -- (cr3);
	\draw [->, dashed, red] (vr22) -- (cr3);
	\draw [->, thick, blue] (vt1) -- (cr11);
	\draw [->, thick, blue] (a) -- (cr1m);
	\draw [->, thick, blue] (cr1m) -- (vr2m);
	\draw [->, dashed, red] (vr2m) -- (cr3);
	\draw [->, thick, blue] (cr3) -- (7, 0);
\end{tikzpicture}
	\caption{The Construction of the Clause-Variable Gadget. Here, weights on \textcolor{red}{dashed} and \textcolor{blue}{thick} edges are \textcolor{red}{$m+n$} and \textcolor{blue}{$1$}, respectively.}
	\label{FigureClauseVariableGadget}
\end{figure}
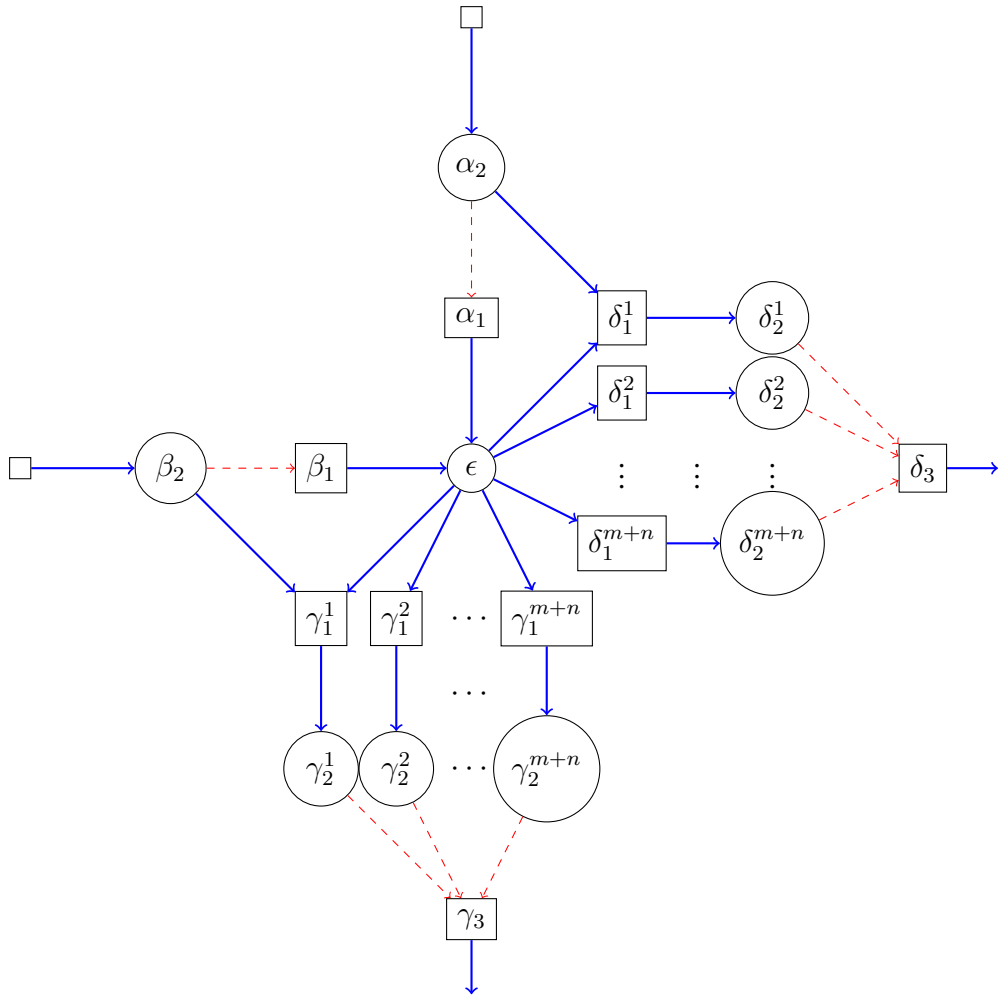

We now describe the insertion of the clause-variable gadget in more detail\footnote{When an arc has more than one crossing point, we repeat the same operation, replacing $v_{j_{k}}$ and $v_{j_{k}'}$ as $\gamma_{3}$ and $\delta_{3}$ appropriately}.
Suppose that arcs from $v_{j_k}$ to $\ell_{j_{k}}$ and from $v_{j_{k}'}$ to $\ell_{j_{k}'}$ intersect. Without loss of generality, we assume that the arc from $v_{j_{k}}$ to $\ell_{j_{k}}$ crosses vertically.
Then, in addition to the arcs on the clause-variable gadgets, we create the following four arcs $(v_{j_{k}}, \alpha_{2})$, $(v_{j_{k}'}, \beta_{2})$, $(\gamma_{3}, \ell_{j_{k}})$, and $(\delta_{3}, \ell_{j_{k}'})$.

By the construction, it is obvious that the clause-variable gadget is planar.
Hence, we can convert the digraph $G_{\phi}$ to the planar digraph by inserting the clause variable gadget into each crossing point. Naturally, our reconstruction can be done in polynomial time since the clause-variable gadget has at most $4m + 4n + 6$ vertices, and the number of crossing points is bounded by some polynomial.

We denote by $H_{\phi}$ the planar digraph constructed from $G_{\phi}$ by the above procedure.
What remains is to prove that $H_{\phi}$ has an undominated out-regular subgraph if and only if the Boolean formula $\phi$ has a satisfying assignment.
\begin{lemma}
\label{LemmaPlanarDiGraph}
	The graph $G_{\phi}$ has an undominated out-regular subgraph if and only if the graph $H_{\phi}$ has an undominated out-regular subgraph.
\end{lemma}
\begin{proof}
	First, we assume that $G_{\phi}$ has an undominated out-regular subgraph $G_{\phi}[S]$, i.e., the subset of vertices $S \subseteq V(G_{\phi})$ induces an undominated out-regular subgraph.
	Then, we show how to construct a subset $T$ of vertices on $H_{\phi}$ that induces the undominated out-regular subgraph $H_{\phi}$. We define the subset $T$ by performing the procedure shown in Algorithm \ref{Algorithm4ConstructingSubsetT}.
	
	\begin{algorithm}
	\caption{A procedure to construct the subset $T$}
	\label{Algorithm4ConstructingSubsetT}
		\begin{algorithmic}[1]
			\State We initialize $T$ with $S$.
			\ForAll {an arc $( v_{j_{k}}, \ell_{i}) \in E(G_{\phi}[S])$ that crosses other arcs on $E(G_{\phi}[S])$}
				\State we add some vertices as follows:
				\If {the path passes through the clause-variable gadget horizontally}
					\State we add the vertices $\beta_{2}$, $\beta_{1}$, $\epsilon$, $\delta_{1}^{m+n}$, $\delta_{2}^{m+n}$, and $\delta_{3}$ to $T$
				\EndIf
				\If {the path passes through the clause-variable gadget vertically}
					\State we add the vertices $\alpha_{2}$, $\alpha_{1}$, $\epsilon$, $\gamma_{1}^{m+n}$, $\gamma_{2}^{m+n}$, and $\gamma_{3}$ to $T$
				\EndIf
			\EndFor
			\ForAll {$\epsilon \in T$}
				\State Set $\kappa = 1$
				\While {$\sum_{u \in \Gamma_{T}^{+}(\epsilon)} w(\epsilon, u) < m + n$}
					\If {$\beta_{1} \in T$}
						\State we add two vertices $\delta_{1}^{m + n - \kappa}$ and $\delta_{2}^{m + n - \kappa}$ to $T$
					\EndIf
					\If {$\alpha_{1} \in T$ and $\sum_{u \in \Gamma_{T}^{+}(\epsilon)} w(\epsilon, u) < m + n$}
						\State we add two vertices $\gamma_{1}^{m + n - \kappa}$ and $\gamma_{2}^{m + n - \kappa}$ to $T$
					\EndIf
					\State $\kappa \leftarrow \kappa + 1$
				\EndWhile
			\EndFor
		\end{algorithmic}	
	\end{algorithm}

	What remains is to prove that the digraph $H_{\phi}[T]$ is an undominated out-regular subgraph.
	By the construction of $H_{\phi}$, the vertices that can dominate $H_{\phi}[T]$ are the variable coordinating vertices $z_{i_2}$ for $i \in [n]$, the vertices $\alpha_{2}$ and $\beta_{2}$ that are on the clause-variable gadget.
	A mentioned in the proof of Theorem \ref{TheoremSATiffUndominatedOutregular}, any variable coordinating vertex $z_{i_2}$ does not dominate the digraph $H_{\phi}[T]$ since at most one of variable vertices $x_{i}$ and $\bar{x}_{i}$ belongs to $S$.
	The vertex $\alpha_{2}$ does not dominate $H_{\phi}[T]$ since at most one of the vertices $\alpha_{1}$ and $\delta_{1}^{1}$ belongs to $T$. Similarly, we see that the vertex $\beta_{2}$ is not a dominating vertex.
	Hence, $H_{\phi}[T]$ is an undominated subgraph. It is straightforward to see that the total weight on out-neighbors is equal to $1$ or $m + n$ for every vertex in $T$. Therefore, $H_{\phi}[T]$ is out-regular. From the above argument, the digraph $H_{\phi}$ has an undominated out-regular subgraph.
	
	Next, we assume that the graph $H_{\phi}$ has an undominated out-regular subgraph $H_{\phi}[T]$, i.e., the subset $T$ induces an undominated out-regular subgraph.
	Then, we show how to construct the subset $S$ of vertices on $G_{\phi}$ that induces an undominated out-regular subgraph. We construct the subgraph $S$ by performing the procedure shown in Algorithm \ref{AlgorithmProcedure4ConstructionT}.
	
	\begin{algorithm}
	\caption{A procedure to construct the subset $S$}
	\label{AlgorithmProcedure4ConstructionT}
		\begin{algorithmic}[1]
			\State We initialize $S$ with $\{ a \} \cup \{ C_{j} ~;~ j \in [m] \}$.
			\ForAll {$i \in [n]$}
				\If {$x_{i} \in T$}
					\State we add the vertices $x_{i}$ and $y_{i}$ to $S$
				\EndIf
				\If {$\bar{x}_{i} \in T$}
					\State we add the vertices $\bar{x}_{i}$ and $\bar{y}_{i}$ to $S$
				\EndIf
			\EndFor
			\ForAll {$ j \in [m]$}
				\If {$u_{j_{k}}$, $v_{j_{k}} \in T$}
					\State we add the vertices $u_{j_{k}}$ and $v_{j_{k}}$ to $S$
				\EndIf
			\EndFor
		\end{algorithmic}
	\end{algorithm}
	
	What remains is to prove that the graph $G_{\phi}[S]$ is an undominated out-regular subgraph.
	First, we observe the next fact: 
	\begin{itemize}
		\item If a clause coordinating vertex $v_{j_k}$ belongs to $T$, there is a variable vertex $x_{i} \in T$ such that the literal $\ell_{j_k}$ that is contained in the clause $C_{j}$ corresponds to $x_{i}$.
	\end{itemize}
	Since the subset $T$ induces the undominated out-regular subgraph $H_{\phi}[T]$, the out-neighbor of $v_{j_k}$ is contained in $T$ whenever $v_{j_k}$ belongs to $T$. Note that a clause coordinating vertex $v_{j_k}$ has only one out-neighbor.
	If the arc outgoing from $v_{j_k}$ does not cross any other arcs, the above fact immediately follows.
	Suppose that the arc outgoing from $v_{j_k}$ crosses other arcs. Thus, the out-neighbor of $v_{j_k}$ is either $\alpha_{2}$ or $\beta_{2}$ on the clause-variable gadget. Without loss of generality, we assume that the out-neighbor of $v_{j_k}$ is $\alpha_{2}$ (the other case follows from the same argument).
	Since again the subset $T$ induces the undominated out-regular subgraph $H_{\phi}[T]$,the vertices $\alpha_{1}$ and $\epsilon$ are contained in $T$. Then, we can see that at least one of vertices $\gamma_{1}^{2}, \gamma_{1}^{3}, \dots, \gamma_{1}^{m+n}$ belongs to $T$ because the vertex $\delta_{1}^{1}$ is not in the subset $T$. If $\delta_{1}^{1}$ is contained in $T$, then the total weight on out-neighbors of $\alpha_{2}$ is not equal to $m + n$, which contradicts the out-regularity of $H_{\phi}[T]$.
	The vertex $\gamma_{3}$ is also contained in $T$ since the subset $T$ contains the vertices $\gamma_{1}^{\kappa}$ and $\gamma_{2}^{\kappa}$ for some $\kappa$.
	By repeating this argument, we see that the above fact follows.
	
	Next, we show that the digraph $G_{\phi}[S]$ is out-regular.
	For each vertex in $S \cap \mathcal{V}_{R}$, the total weight on its out-neighbors is equal to $1$. Furthermore, for each vertex in $S \cap \mathcal{V}_{C}$, the total weight on its out-neighbors is equal to $m + n$. Hence, the induced subgraph $G_{\phi}[S]$ is out-regular.
	
	Finally, we show that the digraph $G_{\phi}[S]$ is undominated.
	For each $i \in [n]$, at most one of variable vertices $x_{i}$ and $\bar{x}_{i}$ is contained in $T$. If not, then the variable coordinating vertex $z_{i_2}$ dominates the digraph $G_{\phi}[S]$. This fact implies that the digraph $G_{\phi}[S]$ is an undominated subgraph.
	From the above argument, the subset $S$ induces an undominated out-regular subgraph.
\end{proof}

%\section{Conclusion}
%\label{SectionConclusion}
%We have studied the complexity of verifying the existence of uniform Nash equilibria for a $\langle \rho_{r}, \rho_{c} \rangle$-planar bimatrix game for some positive integers $\rho_{r}$ and $\rho_{c}$. 
%We have proven that the problem of computing a uniform Nash equilibrium of a $\langle 2, 2 \rangle$-planar bimatrix game is also $\NP$-complete.
%
%This paper left an open question worth considering: How hard to determine the existence of uniform Nash equilibria for a $\langle 1, 2 \rangle$-planar bimatrix game?
%In our opinion, we need to reduce another $\NP$-hard problem to \PlanarBimatrixGame{1}{2} to prove the $\NP$-hardness of that problem.

\section*{Acknowledgments}
This work was supported by JSPS KAKENHI Grant Numbers JP21J10845 and JP20H05795.

% ---- Bibliography ----
\printbibliography[heading = bibintoc]

\end{document}